\documentclass{mgiartcl}
\pdfoutput=1

\usepackage{tikz}
\usetikzlibrary{arrows}
\usetikzlibrary{shapes}
\tikzset{every picture/.style={>=stealth,bend angle=20}}
\tikzset{every label/.style={font=\scriptsize}}
\tikzset{every node/.style={font=\scriptsize}}
\tikzset{play/.style={circle,draw,minimum size=#1}}
\tikzset{play/.default=0.625cm}
\tikzset{prob/.style={diamond,draw,minimum size=#1}}
\tikzset{prob/.default=0.7cm}
\tikzset{end/.style={rectangle,draw,minimum size=#1}}
\tikzset{end/.default=0.55cm}

\title{The Complexity of Nash Equilibria in\\ Simple Stochastic Multiplayer
Games}
\author{Michael Ummels\addr{1} and Dominik Wojtczak\addr{2}\fnmsep\addr{3}}
\address{RWTH Aachen University, Germany \\
\email{ummels@logic.rwth-aachen.de}}
\address{CWI, Amsterdam, The Netherlands \\
\email{d.k.wojtczak@cwi.nl}}
\address{University of Edinburgh, UK}

\DeclareMathOperator{\true}{true}
\DeclareMathOperator{\false}{false}

\DeclareMathOperator{\Prob}{Pr}
\DeclareMathOperator{\Reach}{Reach}

\DeclareMathOperator{\init}{init}
\DeclareMathOperator{\inc}{inc}
\DeclareMathOperator{\dec}{dec}
\DeclareMathOperator{\zero}{zero}

\newcommand{\Inc}[2]{\text{``inc($#1$); goto~$#2$''}}
\newcommand{\Test}[3]{\text{``zero($#1$) ? goto~$#2$ : dec($#1$); goto~$#3$''}}
\newcommand{\Halt}{\text{``halt''}}

\newcommand{\NE}{\mbox{NE}\xspace}
\newcommand{\PureNE}{\mbox{PureNE}\xspace}
\newcommand{\StatNE}{\mbox{StatNE}\xspace}
\newcommand{\PosNE}{\mbox{PosNE}\xspace}
\newcommand{\FinNE}{\mbox{FinNE}\xspace}

\newcommand{\SAT}{\mbox{SAT}\xspace}
\newcommand{\SqrtSum}{\mbox{SqrtSum}\xspace}
\newcommand{\ExThR}{\ensuremath{\mathrm{ExTh}(\mathfrak{R})}\xspace}

\newcommand{\PPAD}{\mbox{PPAD}\xspace}

\newcommand{\cf}{cf.\xspace}
\newcommand{\ea}{\& al.\xspace}
\newcommand{\eg}{e.g.\xspace}
\newcommand{\ie}{i.e.\xspace}

\begin{document}

\maketitle

\begin{abstract}
We analyse the computational complexity of finding Nash equilibria in simple
stochastic multiplayer games. We show that restricting the search space to
equilibria whose payoffs fall into a certain interval may lead to
undecidability.
In particular, we prove that the following problem is undecidable: Given a
game~$\calG$, does there exist a pure-strategy Nash equilibrium of $\calG$
where player~0 wins with probability~1. Moreover, this problem remains
undecidable if it is restricted to strategies with (unbounded) finite memory.
However, if mixed strategies are allowed, decidability remains an open problem.
One way to obtain a provably decidable variant of the problem is to restrict
the strategies to be positional or stationary. For the complexity of these two
problems, we obtain a common lower bound of \NP and upper bounds of \NP and
\PSpace respectively.
\end{abstract}

\section{Introduction}

We study \emph{stochastic games} \cite{NeymanS99} played by multiple
players on a finite, directed graph. Intuitively, a play of such a game evolves
by moving a token along edges of the graph: Each vertex of the graph is either
controlled by one of the players, or it is a \emph{stochastic vertex}. Whenever
the token arrives at a non-stochastic vertex, the player who controls this
vertex must move the token to a successor vertex; when the token arrives at a
stochastic vertex, a fixed probability distribution determines the next vertex.
The play ends when it reaches a terminal vertex, in which case
each player receives a payoff. In the simplest case, which
we discuss here, the possible payoffs of a single play are just 0 and 1
(\ie each player either wins or loses a given play). However,
due to the presence of stochastic vertices, a player's \emph{expected payoff}
(\ie her probability of winning) can be an arbitrary probability.

Stochastic games have been successfully applied in the verification and
synthesis of reactive systems under the influence of random events. Such a 
system is usually modelled as a game between the
system and its environment, where the environment's objective is the complement
of the system's objective: the environment is considered hostile. Therefore,
traditionally, the research in this area has concentrated on two-player games
where each play is won by precisely one of the two players, so-called
\emph{two-player, zero-sum games}. However, the system may comprise of
several components with independent objectives, a situation which is naturally
modelled by a multiplayer game.

The most common interpretation of rational behaviour in multiplayer games is
captured by the notion of a \emph{Nash equilibrium} \cite{Nash50}. In a Nash
equilibrium, no player can improve her payoff by unilaterally switching to a
different strategy. Chatterjee~\ea \cite{ChatterjeeJM04} showed that
any simple stochastic multiplayer game has a Nash equilibrium, and they also
gave an algorithm for computing one. We argue that this is
not satisfactory. Indeed, it can be shown that their algorithm may compute an
equilibrium where all players lose almost surely (\ie receive expected
payoff~0), while there exist other equilibria where all players win almost
surely (\ie receive expected payoff~1).

In applications, one might look for an equilibrium where as many players as
possible win almost surely or where it is guaranteed that the expected payoff
of the equilibrium falls into a certain interval. Formulated as a
decision problem, we want to know, given a $k$-player game~$\calG$ with
initial vertex~$v_0$ and two thresholds $\vec{x},\vec{y}\in [0,1]^k$,
whether $(\calG,v_0)$ has a Nash equilibrium with expected payoff at
least~$\vec{x}$ and at most~$\vec{y}$. This problem, which we call \NE for
short, is a generalisation of Condon's \emph{SSG Problem} \cite{Condon92}
asking whether in a two-player, zero-sum game one of the two players, say
player~0, has a strategy to win the game with probability at
least~$\smallfrac{1}{2}$.

The problem \NE comes in several variants, depending on the type of strategies
one considers: On the one hand, strategies may be \emph{mixed} (allowing
randomisation over actions) or \emph{pure} (not allowing such randomisation).
On the other hand, one can restrict to strategies that use (unbounded or
bounded) finite memory or even to \emph{stationary} ones (strategies that do not
use any memory at all). For the SSG Problem, this distinction is not
meaningful since in a two-player, zero-sum simple stochastic game both players
have an optimal positional (\ie both pure and stationary) strategy
\cite{Condon92}. However, regarding \NE this distinction leads to distinct
decision problems, which have to be analysed separately.

Our main result is that \NE is undecidable if only pure strategies are
considered. In fact, even the following, presumably simpler, problem is
undecidable: Given a game~$\calG$, decide whether there exists a pure Nash
equilibrium where player~0 wins almost surely. Moreover, the problem
remains undecidable if one restricts to pure strategies that use (unbounded)
finite memory. However, for the general case of arbitrary mixed strategies,
decidability remains an open problem.

If one restricts to simpler types of strategies like stationary ones,
the problem becomes provably decidable. In particular, for
positional strategies the problem becomes \NP-complete, and for
arbitrary stationary strategies the problem is \NP-hard but contained in
\PSpace. We also relate the complexity of the latter problem to the complexity
of the infamous \emph{Square Root Sum Problem} (\SqrtSum) by providing a
polynomial-time reduction from \SqrtSum to \NE with the restriction to
stationary strategies. It is a long-standing open problem whether \SqrtSum
falls into the polynomial hierarchy; hence, showing that \NE for stationary
strategies lies inside the polynomial hierarchy would imply a breakthrough
in complexity theory.

Let us remark that our game model is rather restrictive: Firstly, players
receive a payoff only at terminal vertices. In the literature,
a plethora of game models with more complicated modes of winning have been
discussed. In particular, the model of a \emph{stochastic parity game}
\cite{ChatterjeeJH04,Zielonka04} has been investigated thoroughly.
Secondly, our model is \emph{turn-based} (\ie for every non-stochastic vertex
there is only one player who controls this vertex) as opposed to
\emph{concurrent} \cite{AlfaroHK98,AlfaroH00}. The reason that we have chosen
to analyse such a restrictive model is that we are focussing on negative
results. Indeed, all our lower bounds hold for (multiplayer versions of) the
aforementioned models. Moreover, besides Nash equilibria, our negative results
apply to several other solution concepts like subgame perfect
equilibria \cite{Selten65,Ummels06} and secure equilibria
\cite{ChatterjeeHJ06}.

For games with \emph{rewards} on transitions \cite{FilarV97}, the situation
might be different: While our lower bounds can be applied to games with rewards
under the \emph{average reward} or the \emph{total expected reward} criterion,
we leave it as an open question whether this remains true in the case of
discounted rewards.

\subsubsection{Related Work}
Determining the complexity of Nash Equilibria has attracted much
interest in recent years. In particular, a series of papers culminated in the
result that computing a Nash equilibrium of a two-player game in strategic form
is complete for the complexity class \PPAD \cite{DaskalakisGP06,ChenD06}. More
in the spirit of our work, Conitzer and Sandholm \cite{ConitzerS03} showed that
deciding whether there exists a Nash equilibrium in a two-player game in
strategic form where player~0 receives payoff at least $x$ and related
decision problems are all \NP-hard. For infinite games (without stochastic
vertices), (a qualitative version of) the problem \NE was studied in
\cite{Ummels08}. In particular, it was shown that the problem is \NP-complete
for games with parity winning conditions and even in \PTime for games with
B\"uchi winning conditions.

For stochastic games, most results concern the classical SSG problem: Condon
showed that the problem is in $\NP\cap\coNP$ \cite{Condon92}, but it is not
known to be in \PTime. We are only aware of two results that are closely
related to our problem: Firstly, Etessami \ea \cite{EtessamiKVY08} investigated
Markov decision processes with, \eg, multiple reachability objectives. Such a
system can be viewed as a stochastic multiplayer game where all non-stochastic
vertices are controlled by one single player. Under this interpretation, one of
their results states that \NE is decidable in polynomial time for such games.
Secondly, Chatterjee \ea \cite{ChatterjeeJM04} showed that the problem of
deciding whether a (concurrent) stochastic game with reachability objectives
has a positional-strategy Nash equilibrium with payoff at least $\vec{x}$ is
\NP-complete. We sharpen their hardness result by showing
that the problem remains \NP-hard when it is restricted to games with only
three players (as opposed to an unbounded number of players) where,
additionally, payoffs are assigned at terminal vertices only (\cf
\cref{thm:np-hardness} and the subsequent remark).

\section{Simple stochastic multiplayer games}

The model of a \emph{(two-player, zero-sum) simple stochastic game}, introduced
by Condon~\cite{Condon92}, easily generalises to the multiplayer case:
Formally, we define a \emph{simple stochastic multiplayer game (SSMG)} as a
tuple $\calG=(\Pi,V,(V_i)_{i\in\Pi},\Delta,(F_i)_{i\in\Pi})$ such that:
\begin{itemize}
 \item $\Pi$ is a finite set of \emph{players}
  (usually $\Pi=\{0,1,\dots,k-1\}$);
 \item $V$ is a finite set of \emph{vertices};
 \item $V_i\subseteq V$ and $V_i\cap V_j=\emptyset$ for each $i\not=j\in\Pi$;
 \item $\Delta\subseteq V\times([0,1]\cup\{\bot\})\times V$ is the
  \emph{transition relation};
 \item $F_i\subseteq V$ for each $i\in\Pi$.
\end{itemize}
We call a vertex $v\in V_i$ \emph{controlled by player~$i$} and a vertex
that is not contained in any of the sets $V_i$ a \emph{stochastic} vertex.
We require that a transition is labelled by a probability iff
it originates in a stochastic vertex: If $(v,p,w)\in\Delta$ then $p\in
[0,1]$ if $v$ is a stochastic vertex and $p=\bot$ if $v\in V_i$ for some $i\in
\Pi$. Moreover, for each pair of a stochastic vertex~$v$ and an arbitrary
vertex~$w$, we require that there exists precisely one $p\in[0,1]$ such that
$(v,p,w)\in\Delta$. For computational purposes, we require additionally that
all these probabilities are rational.

For a given vertex $v\in V$, we denote the set of all $w\in
V$ such that there exists $p\in(0,1]\cup\{\bot\}$ with $(v,p,w)\in\Delta$ by
$v\Delta$.
For technical reasons, we require that $v\Delta\not=\emptyset$ for all
$v\in V$. Moreover, for each stochastic vertex $v$, the outgoing probabilities
must sum up to 1: $\sum_{(p,w):(v,p,w)\in \Delta} p=1$.
Finally, we require that each vertex~$v$ that lies in one of the sets $F_i$ is a
\emph{terminal (sink) vertex}: $v\Delta=\{v\}$. So if $F$ is the set
of all terminal vertices, then $F_i\subseteq F$ for each $i\in\Pi$.
%A \emph{play} is a maximal path in $\calG$, \ie a sequence $\pi=\pi(0)\pi(1)
%\dots$ such that either $\pi$ is infinite or $\pi(n)\in F$ for some $n$ (and
%hence the play ends in $\pi(n)$.
%A finite play $\pi=\pi(0)\dots\pi(n)$ is
%\emph{won} by player~$i$ if $\pi(n)\in F_i$. An infinite play is not won by
%any player. We denote by $0<|\pi|\leq\omega$ the length of a play $\pi$.

A \emph{(mixed) strategy of player~$i$ in $\calG$} is a mapping
$\sigma:V^*V_i\to\calD(V)$ assigning to each possible \emph{history}
$xv\in V^* V_i$ of vertices ending in a vertex controlled by player~$i$ a
(discrete) probability distribution over $V$ such that $\sigma(xv)(w)>0$ only
if $(v,\bot,w)\in\Delta$. Instead of $\sigma(xv)(w)$, we usually write
$\sigma(w\mid xv)$.
%We say that a play $\pi$ of $\calG$ is \emph{consistent} with a strategy
%$\sigma$ of player~$i$ if $\sigma(\pi(k+1)\mid \pi(0)\dots\pi(k))>0$ for all
%$k<|\pi|-1$ with $\pi(k)\in V_i$.
A \emph{(mixed) strategy profile of $\calG$} is a tuple
$\vec{\sigma}=(\sigma_i)_{i\in\Pi}$ where $\sigma_i$ is a strategy of
player~$i$ in $\calG$.
Given a strategy profile $\vec{\sigma}=(\sigma_j)_{j\in\Pi}$ and a
strategy~$\tau$ of player~$i$, we denote by $(\vec{\sigma}_{-i},
\tau)$ the strategy profile resulting from $\vec{\sigma}$ by replacing
$\sigma_i$ with $\tau$.

A strategy $\sigma$ of player~$i$ is called \emph{pure}
if for each $xv\in V^*V_i$ there exists $w\in v\Delta$ with $\sigma(w\mid xv)
=1$. Note that a pure strategy of player~$i$ can be identified with a function
$\sigma:V^*V_i \to V$. A strategy profile $\vec{\sigma}=(\sigma_i)_{i\in\Pi}$
is called \emph{pure} if each $\sigma_i$ is pure.

A strategy $\sigma$ of player~$i$ in $\calG$ is called \emph{stationary}
if $\sigma$ depends only on the current vertex: $\sigma(xv)=\sigma(v)$
for all $xv\in V^*V_i$. Hence, a stationary strategy of player~$i$ can be
identified with a function $\sigma:V_i\to\calD(V)$.
A strategy profile $\vec{\sigma}=(\sigma_i)_{i\in\Pi}$ of $\calG$ is called
\emph{stationary} if each $\sigma_i$ is stationary.

We call a pure, stationary strategy a \emph{positional strategy} and
a strategy profile consisting of positional strategies only a
\emph{positional strategy profile}. Clearly, a positional strategy of
player~$i$ can be identified with a function $\sigma:V_i\to V$.
More generally, a pure strategy $\sigma$ is called \emph{finite-state} if it
can be implemented by a finite
automaton with output or, equivalently, if the equivalence relation
$\mathord{\sim}\subseteq V^*\times V^*$ defined by $x\sim y$ if $\sigma(xz)=
\sigma(yz)$ for all $z\in V^*V_i$ has only finitely many equivalence
classes.\footnote{In general, this definition is applicable to mixed strategies
as well, but for this paper we will identify finite-state strategies with pure
finite-state strategies.} Finally, a \emph{finite-state strategy profile} is a
profile consisting of finite-state strategies only.

It is sometimes convenient to designate an initial vertex $v_0\in V$ of the
game. We call the tuple $(\mathcal{G},v_0)$ an \emph{initialised SSMG}.
%A \emph{play (history) of $(\mathcal{G},v_0)$} is a play (history) of
%$\mathcal{G}$ starting with $v_0$.
A strategy (strategy profile)
of $(\mathcal{G},v_0)$ is just a strategy (strategy profile) of $\mathcal{G}$.
In the following, we will use the abbreviation SSMG also for initialised SSMGs.
It should always be clear from the context if the game is initialised or not.

Given an SSMG $(\calG,v_0)$ and a strategy profile $\vec{\sigma}=
(\sigma_i)_{i\in\Pi}$, the \emph{conditional probability of $w\in V$ given the
history $xv\in V^*V$} is the number $\sigma_i(w\mid xv)$ if $v\in V_i$ and the
unique $p\in [0,1]$ such that $(v,p,w)\in\Delta$ if $v$ is a stochastic vertex.
We abuse notation and denote this probability by $\vec{\sigma}(w\mid xv)$. The
probabilities $\vec{\sigma}(w\mid xv)$ induce a probability measure on the
space $V^\omega$ in the following way: The probability of a basic open set
$v_1\dots v_k\cdot V^\omega$ is 0 if $v_1\not=v_0$ and the product of the
probabilities $\vec{\sigma}(v_j\mid v_1\dots v_{j-1})$ for $j=2,\dots,k$
otherwise. It is a classical result of measure theory that this extends to a
unique probability measure assigning a probability to every Borel subset of
$V^\omega$, which we denote by $\Prob_{v_0}^{\vec{\sigma}}$.

For a set $U\subseteq V$, let $\Reach(U):=V^*\cdot U\cdot V^\omega$. We
are mainly interested in the probabilities $p_i:=\Prob_{v_0}^{\vec{\sigma}}
(\Reach(F_i))$ of reaching the sets $F_i$. We call $p_i$ the \emph{(expected)
payoff of $\vec{\sigma}$ for player~$i$} and the vector $(p_i)_{i\in\Pi}$ the
\emph{(expected) payoff of $\vec{\sigma}$}.
Another way to define these probabilities is via the \emph{Markov chain}
$\calG^{\vec{\sigma}}$ which is defined as follows: The state set of
$\calG^{\vec{\sigma}}$ is $V^+$ (the set of all nonempty sequences of
vertices), and the probability of going from state~$xv$ to
state~$xvw$ ($x\in V^*$, $v,w\in V$) is equal to $\vec{\sigma}(w\mid xv)$. Then
the expected payoff of $\vec{\sigma}$ for player~$i$ can be computed as the
probability of reaching a state~$xv$ with $v\in F_i$ from state~$v_0$ in
$\calG^{\vec{\sigma}}$.

\subsubsection{Drawing an SSMG}
When drawing an SSMG as a graph, we will use the
following conventions: The initial vertex is marked by an incoming edge that
has no source vertex. Vertices that are controlled by a player are depicted as
circles, where the player who controls a vertex is given by the label next
to it. Stochastic vertices are depicted as diamonds, where the
transition probabilities are given by the labels on its outgoing edges (the
default being~\smallfrac{1}{2}). Finally, terminal
vertices are generally represented by their associated payoff vector. In fact,
we allow arbitrary vectors of rational probabilities as payoffs. This does not
increase the power of the model since such a payoff vector can easily be
realised by an SSMG consisting of stochastic and terminal vertices only.

\section{Nash equilibria}

To capture rational behaviour of (selfish) players, John Nash \cite{Nash50}
introduced the notion of, what is now called, a \emph{Nash equilibrium}. Formally,
given a strategy profile~$\vec{\sigma}$, a strategy~$\tau$ of player~$i$
is called a \emph{best response to $\vec{\sigma}$} if $\tau$ maximises
the expected payoff of player~$i$:
$\Prob_{v_0}^{(\vec{\sigma}_{-i},\tau')}(\Reach(F_i))\leq
\Prob_{v_0}^{(\vec{\sigma}_{-i},\tau)}(\Reach(F_i))$ for all strategies $\tau'$
of player~$i$. A Nash equilibrium is a strategy profile $\vec{\sigma}=
(\sigma_i)_{i\in\Pi}$ such that each $\sigma_i$ is a best response to
$\vec{\sigma}$. Hence, in a Nash equilibrium no player can improve
her payoff by (unilaterally) switching to a different strategy.

Previous research on algorithms for finding Nash equilibria in infinite games
has focused on computing \emph{some} Nash equilibrium \cite{ChatterjeeJM04}.
However, a game may have several Nash equilibria with different payoffs, and
one might not be interested in \emph{any} Nash equilibrium but in one whose
payoff fulfils certain requirements. For example, one might look for a
Nash equilibrium where certain players win almost surely while certain others
lose almost surely. This idea leads us to the following decision problem, which
we call \NE:\footnote{In the definition of \NE, the ordering $\leq$ is applied
componentwise.}
\begin{quote}
Given an SSMG $(\calG,v_0)$ and thresholds
$\vec{x},\vec{y}\in[0,1]^\Pi$, decide whether there exists a
Nash equilibrium of $(\calG,v_0)$ with payoff $\geq\vec{x}$ and
$\leq\vec{y}$.
\end{quote}
For computational purposes, we assume that the thresholds $\vec{x}$
and $\vec{y}$ are vectors of rational numbers.
A variant of the problem which omits the thresholds just asks about a Nash
equilibrium where some distinguished player, say player~0, wins with
probability~1:
\begin{quote}
Given an SSMG $(\calG,v_0)$, decide whether there exists a Nash equilibrium of
$(\calG,v_0)$ where player~0 wins almost surely.
\end{quote}
Clearly, every instance of the threshold-free variant can easily be turned into
an instance of \NE (by adding the thresholds $\vec{x}=(1,0,\ldots,0)$ and
$\vec{y}=(1,\ldots,1)$). Hence, \NE is, a priori, more general than its
threshold-free variant.

Our main concern in this paper are variants of \NE where we restrict the
type of strategies that are allowed in the definition of the problem: Let
\PureNE, \FinNE, \StatNE and \PosNE be the problems that arise from \NE by
restricting the desired Nash equilibrium to consist of pure strategies,
finite-state strategies, stationary strategies and positional strategies,
respectively. In the rest of this paper, we are going to prove upper and
lower bounds on the complexity of these problems, where all lower
bounds hold for the threshold-free variants, too.

Our first observation is that neither stationary nor pure strategies are
sufficient to implement any Nash equilibrium, even if we are only interested in
whether a player wins or loses almost surely in the
Nash equilibrium. Together with a result from \cref{section:undecidability}
(namely \cref{prop:inf-mem}), this demonstrates that the problems
\NE, \PureNE, \FinNE, \StatNE, and \PosNE are pairwise distinct problems,
which have to be analysed separately.

\begin{proposition}\label{prop:stationary-not-enough}
There exists an SSMG that has a finite-state Nash equilibrium where player~0
wins almost surely but that has no stationary Nash equilibrium where player~0
wins with positive probability.
\end{proposition}
\begin{proof}
Consider the game~$\calG$ depicted in \cref{figure:no-stat-nash} played by
three players 0, 1 and 2 (with payoffs in this order).
Obviously, the following
finite-state strategy profile is a Nash equilibrium where player~0
wins almost surely:
Player~1 plays from vertex~$v_2$ to vertex~$v_3$ at the first visit of $v_2$
but leaves the game immediately (by playing to the neighbouring terminal
vertex) at all subsequent visits to $v_2$; from vertex~$v_0$ player~1 plays to
$v_1$; player~2 plays from vertex~$v_3$ to vertex~$v_4$ at the first visit of
$v_3$ but leaves the game immediately at all subsequent visits to $v_3$; from
vertex~$v_1$ player~2 plays to $v_2$.
\begin{figure}[ht]
\centering
\begin{tikzpicture}[x=1.5cm,y=1.3cm,->]
\node (init) at (-0.7,0) {};
\node (0) at (0,0) [play,label={above:1}] {$v_0$};
\node (pay0) at (0,-1) {$(0,\smallfrac{1}{2},0)$};
\node (1) at (1,0) [play,label={above:2}] {$v_1$};
\node (pay1) at (1,-1) {$(0,0,\smallfrac{1}{2})$};
\node (2) at (2,0) [play,label={above:1}] {$v_2$};
\node (pay2) at (2,-1) {$(1,0,1)$};
\node (3) at (2.75,1) [play,label={left:2}] {$v_3$};
\node (pay3) at (2.75,2) {$(0,0,0)$};
\node (4) at (3.5,0) [prob] {$v_4$};
\node (pay4) at (4.7,0) {$(1,1,0)$};

\draw (init) to (0);
\draw (0) to (1); \draw (0) to (pay0);
\draw (1) to (2); \draw (1) to (pay1);
\draw (2) to (3); \draw (2) to (pay2);
\draw (3) to (4); \draw (3) to (pay3);
\draw (4) to (2); \draw (4) to (pay4);
\end{tikzpicture}
\caption{\label{figure:no-stat-nash} An SSMG with three players}
\end{figure}
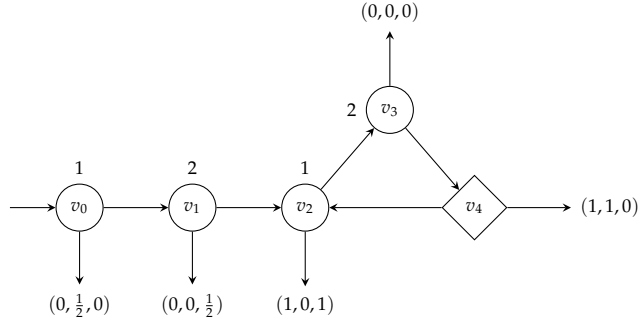

It remains to show that there is no stationary Nash equilibrium of
$(\calG,v_0)$ where player~0 wins with positive probability. Any stationary
Nash equilibrium of $(\calG,v_0)$ where player~0 wins with positive probability
induces a stationary Nash equilibrium of $(\calG,v_2)$ where both players~1 and
2 receive payoff at least \smallfrac{1}{2} since
otherwise one of these players could improve her payoff by changing her
strategy at $v_0$ or $v_1$.
Hence, it suffices to show that $(\calG,v_2)$ has no stationary Nash
equilibrium where both players 1 and 2 receive payoff at least
\smallfrac{1}{2}. Assume there exists such an equilibrium and denote
by $p$ the probability that player~2 plays from $v_3$ to $v_4$. Since
player~1 wins with probability~$>0$, it must be the case that $p>0$. But then,
to have a Nash equilibrium, player~1 must play from $v_2$ to $v_3$ with
probability~1, giving player~2 a payoff of 0, a contradiction.
\end{proof}

\begin{proposition}\label{prop:pure-not-enough}
There exists an SSMG that has a stationary Nash equilibrium where player~0
wins almost surely but that has no pure Nash equilibrium where player~0 wins
with positive probability.
\end{proposition}
\begin{proof}
Consider the game depicted in \cref{figure:no-pure-nash} played by three
players 0, 1 and 2 (with payoffs given in this order).
Clearly, the stationary strategy profile where from vertex~$v_2$ player~0
selects both outgoing edges with probability~\smallfrac{1}{2}\ each, player~1
plays from $v_0$ to $v_1$ and player~2 plays from $v_1$ to $v_2$ is
a Nash equilibrium where player~0 wins almost surely. However, for any pure
strategy profile where player~0 wins almost surely, either player~1 or
player~2 receives payoff~0 and could improve her payoff by switching her
strategy at $v_0$ or $v_1$ respectively.
\begin{figure}[ht]
\centering
\begin{tikzpicture}[x=1.5cm,y=1.3cm,->]
\node (init) at (-0.7,0) {};
\node (0) at (0,0) [play,label={above:1}] {$v_0$};
\node (1) at (1,0) [play,label={above:2}] {$v_1$};
\node (2) at (2,0) [play,label={above:0}] {$v_2$};
\node (3) at (2.75,-0.75) {$(1,1,0)$};
\node (4) at (2.75,0.75) {$(1,0,1)$};
\node (5) at (0,-1) {$(0,\smallfrac{1}{2},0)$};
\node (6) at (1,-1) {$(0,0,\smallfrac{1}{2})$};
\draw (init) -- (0);
\draw (0) -- (1);
\draw (0) -- (5);
\draw (1) -- (2);
\draw (1) -- (6);
\draw (2) -- (3);
\draw (2) -- (4);
\end{tikzpicture}
\caption{\label{figure:no-pure-nash} Another SSMG with three players}
\end{figure}
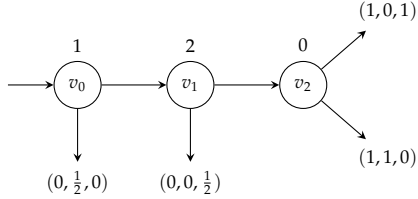
\end{proof}

\section{Decidable variants of \NE}

\subsection{Upper bounds}

In this section, we show that the problems \PosNE and \StatNE are contained in
the complexity classes \NP and \PSpace respectively.

\begin{theorem}\label{thm:pos-ne-np}
\PosNE is in \NP.
\end{theorem}
\begin{proof}
Let $(\calG,v_0)$ be an SSMG. Any positional strategy profile of $\calG$ can be
identified with a mapping $\vec{\sigma}:\bigcup_{i\in\Pi} V_i\to V$ such that
$(v,\bot,\vec{\sigma}(v))\in\Delta$ for each non-stochastic vertex~$v$, an
object whose size is linear in the size of $\calG$. To prove that \PosNE is
in \NP, it suffices to show that we can check in polynomial time whether such
a mapping $\vec{\sigma}$ constitutes a Nash equilibrium whose payoff
lies in between the given thresholds $\vec{x}$ and $\vec{y}$.

First, we need to compute the payoff of $\vec{\sigma}$. Let
$z^i_v:=\Prob_v^{\vec{\sigma}}(\Reach(F_i))$ denote the expected payoff of
$\vec{\sigma}$ for player~$i$ in $(\calG,v)$, and let $\vec{z}^i=
(z^i_v)_{v\in V}$. It is a well-known result of the theory of Markov chains
that $\vec{z}^i$ is the optimal solution of the following linear programme:
\begin{align*}
\shortintertext{Minimise $\sum_{v\in V} z^i_v$, subject to:}
z^i_v & \geq 0 && \text{for $v\in V$,} \\
z^i_v &= 1 && \text{for $v \in F_i$,} \\
z^i_v &= \sum_{w \in V} \vec{\sigma}(w\mid v)\cdot z^i_w &&
\text{for $v\in V\setminus F_i$.}
\end{align*}
Once we have computed $\vec{z}^i$, we can check whether $x_i \leq
z^i_{v_0} \leq y_i$; this inequality holds for each player $i\in\Pi$ iff the
payoff of $\vec{\sigma}$ lies in between $\vec{x}$ and $\vec{y}$.

To check whether $\vec{\sigma}$ is a Nash equilibrium, we need to compute the
numbers $\sup_\tau \Prob_{v_0}^{(\vec{\sigma}_{-i},\tau)}(\Reach(F_i))$ (where
$\tau$ ranges over every strategy of player~$i$ in~$\calG$), the
maximal payoff that player~$i$ can achieve when playing against
$\vec{\sigma}_{-i}$. If this payoff is equal to $z^i_{v_0}$, then she
cannot gain anything by unilaterally switching to any other strategy.
From the theory of \emph{Markov decision process} (\cf \cite{Puterman94}),
it is well-known that the desired payoff can be computed by the following
linear programme over the variables $\vec{r}^i=(r^i_v)_{v\in V}$:
\begin{align*}
\shortintertext{Minimise $\sum_{v\in V} r^i_v$, subject to:}
r^i_v & \geq 0 && \text{for $v\in V$,} \\
r^i_v &= 1 && \text{for $v \in F_i$,} \\
%r^i_v &= 0 && \text{for $v \in F\setminus F_i$,} \\
r^i_v &\geq r^i_w && \text{for $v \in V_i$ and $w \in v\Delta$,} \\
r^i_v &= \sum_{w \in V} \vec{\sigma}(w\mid v)\cdot r^i_w &&
\text{for $v\in V\setminus V_i$.}
\end{align*}
To check whether $\vec{\sigma}$ is a Nash equilibrium, it suffices
to compute for each player~$i$ the optimal solution $\vec{r}^i$ and to check
whether $r^i_{v_0}=z^i_{v_0}$.

Since linear programmes can be solved in polynomial time and both programmes
are of size polynomial in the size of the game, all these checks can be carried
out in polynomial time.
\end{proof}

To prove the decidability of \StatNE, we appeal to results established
for the \emph{Existential Theory of the Reals}, \ExThR, the set of all
existential first-order sentences (over the appropriate signature) that hold in
$\frakR:=(\bbR,+,\cdot,0,1,\leq)$. The best known upper bound for the
complexity of the associated decision problem is \PSpace
\cite{Canny88,Renegar92}, which leads to the following theorem.

\begin{theorem}\label{thm:stat-ne-pspace}
\StatNE is in \PSpace.
\end{theorem}
\begin{proof}
Instead of giving a deterministic polynomial-space algorithm for \StatNE, we
give a nondeterministic one. Since $\PSpace=\NPSpace$, this implies that
\StatNE is in \PSpace.
On input $\calG,v_0,\vec{x},\vec{y}$, the algorithm starts by guessing a set
$S\subseteq V\times V$ and proceeds by computing, for
each player~$i$, the set $R_i$ of vertices from where the set $F_i$ is
reachable in the graph $G=(V,S)$, a computation which can be carried out in
polynomial time.
Note that if $S$ is the \emph{support} of a stationary strategy
profile~$\vec{\sigma}$, \ie $S=\{(v,w)\in V\times V:\vec{\sigma}(w\mid v)>0\}$,
then $R_i$ is precisely the set of vertices~$v$ such that
$\Prob_v^{\vec{\sigma}}(\Reach(F_i))>0$.
Finally, the algorithm evaluates an existential first-order
sentence~$\psi$, which can be computed in polynomial time from $(\calG,v_0)$,
$\vec{x}$, $\vec{y}$, $S$ and $(R_i)_{i\in\Pi}$, over $\frakR$ and returns the
answer to this query.

It remains to describe a suitable sentence $\psi$.
Let $\vec{\alpha}=(\alpha_{vw})_{v,w\in V}$, $\vec{r}=(r^i_v)_{i\in\Pi,v\in V}$
and $\vec{z}=(z^i_v)_{i\in\Pi, v\in V}$ be three sets of variables, and let
$V_\ast=\bigcup_{i \in \Pi} V_i$ be the set of all non-stochastic vertices. The
formula
\begin{align*}
\phi(\vec{\alpha})&:=
\bigwedge_{v\in V_\ast} \bigg( \bigwedge_{w\in v\Delta} \alpha_{vw} \geq 0\
\wedge \bigwedge_{\mathmakebox[0.7cm][c]{w\in V\setminus v\Delta}}
\alpha_{vw}=0\wedge \sum_{w \in v\Delta} \alpha_{vw} = 1\bigg)\,\wedge \\
& \quad
\bigwedge_{\mathmakebox[0.7cm][c]{\substack{v\in V\setminus V_\ast \\ w\in V}}}
\alpha_{vw} = p_{vw}\wedge \bigwedge_{\mathmakebox[0.7cm][c]{(v,w)\in S}}
\alpha_{vw} > 0 \wedge
\bigwedge_{\mathmakebox[0.7cm][c]{(v,w)\not\in S}} \alpha_{vw}=0\,,
\end{align*}
where $p_{vw}$ is the unique number such that $(v,p_{vw},w)\in\Delta$,
states that the mapping $\vec{\sigma}:V\to\calD(V)$ defined by
$\vec{\sigma}(w\mid v)=\alpha_{vw}$
constitutes a valid stationary strategy profile of $\calG$ whose support
is~$S$. Provided that $\phi(\vec{\alpha})$ holds in $\frakR$, the formula
\[\eta_i(\vec{\alpha},\vec{z}):=
\bigwedge_{\mathclap{v \in F_i}} z^i_v= 1 \wedge
\bigwedge_{\mathmakebox[0.7cm][c]{v \in V\setminus R_i}} z^i_v= 0 \wedge
\bigwedge_{\mathmakebox[0.7cm][c]{v \in V\setminus F_i}} z^i_v =
\sum_{w \in v\Delta} \alpha_{vw} z^i_w\]
states that $z^i_v=\Prob_v^{\vec{\sigma}}(\Reach(F_i))$ for each $v\in V$,
where $\vec{\sigma}$ is defined as above. Again, this follows from a well-known
results about Markov chains, namely that the vector of the aforementioned
probabilities is the unique solution to the given system of equations. Finally,
the formula
\[\theta_i(\vec{\alpha},\vec{r}):=
\bigwedge_{v\in V} r^i_v\geq 0 \wedge
\bigwedge_{v\in F_i} r^i_v=1 \wedge
\bigwedge_{\substack{v\in V_i \\ w\in v\Delta}} r^i_v\geq r^i_w\wedge
\bigwedge_{\mathmakebox[0.7cm][c]{v\in V\setminus V_i}}
r^i_v=\sum_{w\in v\Delta}\alpha_{vw} r^i_w\]
states that $\vec{r}$ is a solution of the linear programme for computing
the maximal payoff that player~$i$ can achieve when playing against the
strategy profile $\vec{\sigma}_{-i}$. In particular, the formula is fulfilled
if $r^i_v=\sup_{\tau}\Prob^{(\vec{\sigma}_{-i},\tau)}_v(\Reach(F_i))$ (where
$\tau$ ranges over \emph{every} strategy of player~$i$), and
every other solution is greater than this one (in each component).

The desired sentence $\psi$ is the existential closure of the conjunction of
$\phi$ and, for each player~$i$, the formulae $\eta_i$ and $\theta_i$ combined
with formulae stating that player~$i$ cannot improve her payoff and that the
expected payoff for player~$i$ lies in between the given thresholds:
\[\psi:=\exists\vec{\alpha}\,\exists\vec{r}\,\exists\vec{z}\,
\big(\phi(\vec{\alpha})\wedge\bigwedge_{i\in\Pi}(\eta_i(\vec{\alpha},\vec{z})
\wedge\theta_i(\vec{\alpha},\vec{r})\wedge r^i_{v_0}\leq z^i_{v_0}\wedge
x_i\leq z^i_{v_0}\leq y_i)\big)\]
It follows that $\psi$ holds in $\frakR$ iff $(\calG,v_0)$ has a
stationary Nash equilibrium $\vec{\sigma}$ with payoff at least $\vec{x}$ and
at most $\vec{y}$ whose support is~$S$. Consequently, the algorithm is correct.
\end{proof}

\subsection{Lower bounds}

Having shown that \PosNE and \StatNE are in \NP and \PSpace respectively,
the natural question arises whether there is a polynomial-time algorithm
for \PosNE or \StatNE. The following theorem shows that this is not the case
(unless, of course, $\PTime=\NP$) since both problems are \NP-hard.
Moreover, both problems are already \NP-hard for games with only two players.

\begin{theorem}\label{thm:np-hardness}
\PosNE and \StatNE are \NP-hard, even for games with only two players.
\end{theorem}
\begin{proof}
The proof is by reduction from \SAT. Let $\phi=C_1\wedge\dots\wedge C_m$ be
a formula in conjunctive normal form over propositional variables $X_1,\dots,
X_n$. Our aim is to construct a two-player SSMG $(\calG_\phi,v_0)$ such that
the following statements are equivalent:
\begin{enumerate}
 \item\label{item:reduction-1} $\phi$ is satisfiable;
 \item\label{item:reduction-2} $(\calG_\phi,v_0)$ has a positional Nash
  equilibrium with payoff $(1,\smallfrac{1}{2})$;
 \item\label{item:reduction-3} $(\calG_\phi,v_0)$ has a stationary Nash
  equilibrium with payoff $(1,\smallfrac{1}{2})$.
\end{enumerate}
Provided that the game can be constructed in polynomial time, the equivalence
of \ref{item:reduction-1}.\ and \ref{item:reduction-2}.\ establishes a
polynomial-time reduction from \SAT to \PosNE, whereas the equivalence of
\ref{item:reduction-1}.\ and \ref{item:reduction-3}.\ establishes one from \SAT
to \StatNE.
The game $\calG_\phi$ is depicted in \cref{fig:reduction-np} and played by
players 0 and 1. The game proceeds
from the initial vertex~$v_0$ to $X_i$ or $\vec{X_i}$ with
probability~$\frac{1}{2^{i+1}}$ each, and there is an edge from vertex $C_j$ to
vertex $X_i$ or $\vec{X_i}$ iff $X_i$ or $\neg X_i$ respectively
occurs in the clause $C_j$. Also, from $\top$-labelled vertices player~1 can
``leave the game'' by moving to a terminal vertex with
payoff~$(0,1)$. Obviously, the game $\calG_\phi$ can be constructed
from $\phi$ in polynomial time. It remains to show that
\ref{item:reduction-1}.--\ref{item:reduction-3}.\ are equivalent.
\begin{figure}[htb]
\centering
\begin{tikzpicture}[x=1.5cm,y=1.3cm,->,bend angle=15]
\tikzset{play/.style={circle,draw,minimum size=0.7cm}}
\tikzset{prob/.style={diamond,draw,minimum size=0.75cm}}

\node (init) at (-0.7,0) {};
\node (start) at (0,0) [prob] {$v_0$};
\node (phi) at (1,0) [prob] {$\phi$};
\node (end) at (0,3.5) {$(1,0)$};

\draw (init) to (start);
\draw (start) to node[above] {$\frac{1}{2^{n+1}}$} (phi);
\draw (start) to node[left] {$\frac{1}{2^{n+1}}$} (end);

\node (c1) at (1.7,2) [play,label={above:1}] {$C_1$};
\node (clauses1) at (1.7,1) {$\vdots$};
\node (clauses2) at (1.7,-1) {$\vdots$};
\node (cm) at (1.7,-2) [play,label={below:1}] {$C_m$};
\node (pay) at (5.5,0) {$(1,1)$};

\draw (phi) to node[right,near start] {$\frac{1}{m+1}$} (c1);
\draw (phi) to node[right,near start] {$\frac{1}{m+1}$} (cm);
\draw (phi) to node[above,very near start] {$\frac{1}{m+1}$} (pay);

\begin{scope}[yshift=1.3cm]
\node (x1) at (3,2.75) [play,label={above:0}] {$X_1$};
\node (x1z) at (3.6,3.55) [play,label={above:1}] {$\top$};
\node (x1b) at (4.5,3) [prob] {};
\node (x1bf) at (4.5,2.25) [prob] {$\bot$};
\node (pay1z) at (5.5,3.75) {$(0,1)$};
\node (pay1) at (5.5,3) {$(1,1)$};
\node (pay1f) at (5.5,2.25) {$(1,0)$};

\node (nx1) at (3,1) [play,label={above:0}] {$\vec{X_1}$};
\node (nx1b) at (4.5,0.75) [prob] {};
\node (nx1bf) at (4.5,1.5) [prob] {$\bot$};
\node (nx1z) at (3.6,0.2) [play,label={below:1}] {$\top$};
\node (npay1) at (5.5,0.75) {$(1,1)$};
\node (npay1f) at (5.5,1.5) {$(1,0)$};
\node (npay1z) at (5.5,0) {$(0,1)$};

\draw (x1) to [bend left=10] (x1z);
\draw (x1z) to (pay1z);
\draw (x1z) [bend left=10] to (x1b);
\draw (x1) to [bend left] (x1bf);
\draw (x1b) to [bend right=30] node[below,near start] {$\frac{1}{2}$} (x1);
\draw (x1bf) to [bend left] node[below] {$\frac{1}{2}$} (x1);
\draw (x1b) to node[above] {$\frac{1}{2}$} (pay1);
\draw (x1bf) to node[above] {$\frac{1}{2}$} (pay1f);

\draw (nx1) to [bend right=10] (nx1z);
\draw (nx1z) to [bend right=10] (nx1b);
\draw (nx1z) to (npay1z);
\draw (nx1) to [bend right] (nx1bf);
\draw (nx1b) to [bend left=30] node[above,near start] {$\frac{1}{2}$} (nx1);
\draw (nx1bf) to [bend right] node[above] {$\frac{1}{2}$} (nx1);
\draw (nx1b) to node[above] {$\frac{1}{2}$} (npay1);
\draw (nx1bf) to node[above] {$\frac{1}{2}$} (npay1f);
\end{scope}

\draw (c1) to [dashed] (x1); \draw (c1) to [dashed] (nx1);
\draw (cm) to [dashed] (x1); \draw (cm) to [dashed] (nx1);

\node (vars1) at (3,0.6) {$\vdots$};
\node (vars2) at (3,-0.6) {$\vdots$};
\node (nodes3) at (4.5,0.6) {$\vdots$};
\node (nodes4) at (4.5,-0.6) {$\vdots$};

\begin{scope}[yshift=-6.175cm]
\node (x2) at (3,2.75) [play,label={below:0}] {$X_n$};
\node (x2z) at (3.6,3.55) [play,label={above:1}] {$\top$};
\node (x2b) at (4.5,3) [prob] {};
\node (x2bf) at (4.5,2.25) [prob] {$\bot$};
\node (pay2z) at (5.5,3.75) {$(0,1)$};
\node (pay2) at (5.5,3) {$(1,1)$};
\node (pay2f) at (5.5,2.25) {$(1,0)$};

\node (nx2) at (3,1) [play,label={below:0}] {$\vec{X_n}$};
\node (nx2b) at (4.5,0.75) [prob] {};
\node (nx2bf) at (4.5,1.5) [prob] {$\bot$};
\node (nx2z) at (3.6,0.2) [play,label={below:1}] {$\top$};
\node (npay2) at (5.5,0.75) {$(1,1)$};
\node (npay2f) at (5.5,1.5) {$(1,0)$};
\node (npay2z) at (5.5,0) {$(0,1)$};

\draw (x2) to [bend left=10] (x2z);
\draw (x2z) to (pay2z);
\draw (x2z) [bend left=10] to (x2b);
\draw (x2) to [bend left] (x2bf);
\draw (x2b) to [bend right=30] node[below,near start] {$\frac{1}{2}$} (x2);
\draw (x2bf) to [bend left] node[below] {$\frac{1}{2}$} (x2);
\draw (x2b) to node[above] {$\frac{1}{2}$} (pay2);
\draw (x2bf) to node[above] {$\frac{1}{2}$} (pay2f);

\draw (nx2) to [bend right=10] (nx2z);
\draw (nx2z) to [bend right=10] (nx2b);
\draw (nx2z) to (npay2z);
\draw (nx2) to [bend right] (nx2bf);
\draw (nx2b) to [bend left=30] node[above,near start] {$\frac{1}{2}$} (nx2);
\draw (nx2bf) to [bend right] node[above] {$\frac{1}{2}$} (nx2);
\draw (nx2b) to node[above] {$\frac{1}{2}$} (npay2);
\draw (nx2bf) to node[above] {$\frac{1}{2}$} (npay2f);
\end{scope}

\draw (c1) to [dashed] (x2); \draw (c1) to [dashed] (nx2);
\draw (cm) to [dashed] (x2); \draw (cm) to [dashed] (nx2);

\draw (start) to [bend left=30] node[left,near start] {$\frac{1}{4}$} (x1);
\draw (start) to [bend left=15] node[above,near start] {$\frac{1}{4}$} (nx1);
\draw (start) to [bend right=15] node[below,near start] {$\frac{1}{2^{n+1}}$} (x2);
\draw (start) to [bend right=30] node[left,near start] {$\frac{1}{2^{n+1}}$} (nx2);
\end{tikzpicture}
\caption{\label{fig:reduction-np} Reducing \SAT to \PosNE and \StatNE.}
\end{figure}
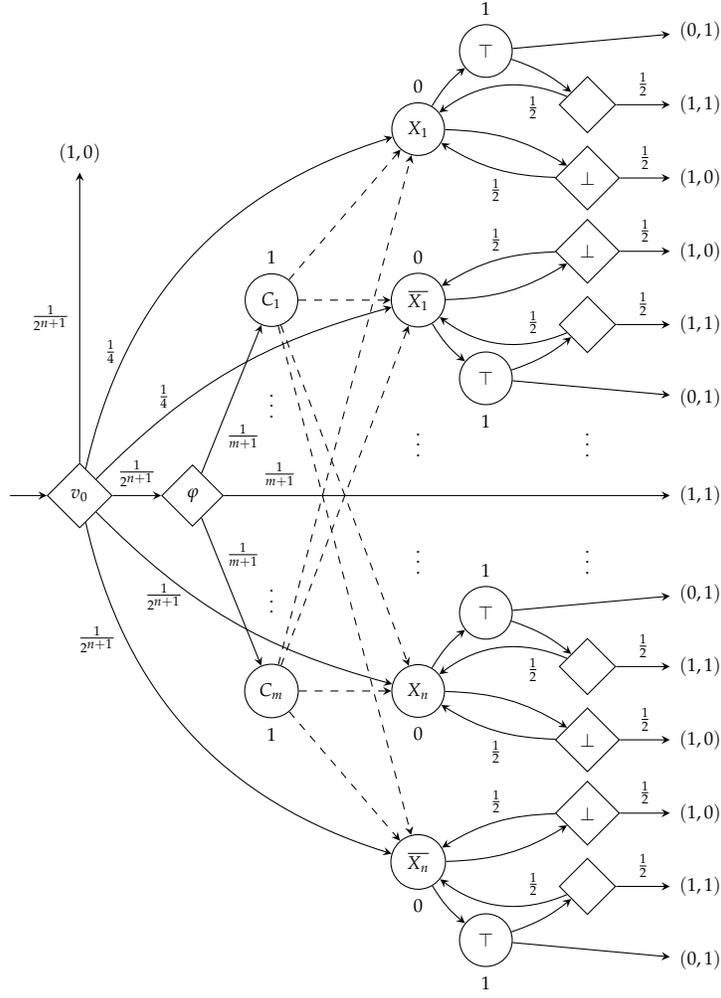

$(1.\Rightarrow 2.)$ Assume that $\alpha:\{X_1,\dots,X_n\}\to\{\true,\false\}$
is a satisfying assignment of $\phi$. In the positional Nash
equilibrium of $(\calG,v_0)$, player~0 moves from a literal $L$ (\ie $L=X_i$
or $L=\vec{X_i}$ for some $i=1,\dots,n$) to the $\top$-labelled vertex iff $L$
is mapped to $\true$ by $\alpha$, and player~1 moves from vertex $C_j$ to a
(fixed) literal~$L$ that is contained in $C_j$ and mapped to $\true$ by
$\alpha$ (which is possible since $\alpha$ is a satisfying assignment). At
$\top$-labelled vertices, player~1 never leaves the game. Obviously, player~0
wins almost surely with this strategy profile. For player~1, the payoff is
\[\frac{1}{2^{n+1}}+\sum_{i=1}^n \frac{1}{2^{i+1}}
=\frac{1}{2^{n+1}}+\frac{1}{2}\big(\sum_{i=1}^n \frac{1}{2^i}\big)
=\frac{1}{2^{n+1}}+\frac{1}{2}\big(1-\frac{1}{2^n}\big)
=\frac{1}{2}\,,\]
where the first summand is the probability of going from the initial vertex
to $\phi$, from where player~1 wins almost surely since from every clause
vertex she plays to a ``true'' literal. Obviously, changing her strategy cannot
give her a better payoff. Therefore, we have a Nash equilibrium.

$(2.\Rightarrow 3.)$ Obvious.

$(3.\Rightarrow 1.)$ Let $\vec{\sigma}=(\sigma_0,\sigma_1)$ be a stationary
Nash equilibrium of $(\calG_\phi,v_0)$ with payoff $(1,\smallfrac{1}{2})$. Our
first aim is to show that $\sigma_0$ is actually a positional strategy. Towards
a contradiction, assume that there exists a literal~$L$ such that $\sigma_0(L)$
assigns probability $0<q<1$ to the neighbouring $\top$-labelled vertex. Since
player~0 wins almost surely, player~1 never leaves the game. Hence, the
expected payoff for player~1 from vertex~$L$ (\ie in the game $(\calG_\phi,L)$)
is precisely~$q$. However, if she left the game at the $\top$-labelled vertex,
she would receive payoff~$\frac{2q}{1+q}>q$. Therefore, $\vec{\sigma}$ is not a
Nash equilibrium, a contradiction.

Knowing that $\sigma_0$ is a positional strategy, we can define a \emph{pseudo
assignment} $\alpha:\{X_1,\neg X_1,\dots,X_n,\neg X_n\}\to
\{\true,\false\}$ by setting $\alpha(L)=\true$ if $\sigma_1$ prescribes
to go from vertex~$L$ to the neighbouring $\top$-labelled vertex. Our next
aim is to show that $\alpha$ is actually an assignment: $\alpha(X_i)=\true
\Leftrightarrow \alpha(\neg X_i)=\false$. To see this, note that we can compute
player~1's expected payoff as follows:
\[\frac{1}{2}=\frac{p}{2^{n+1}}+\sum_{i=1}^n \frac{a_i}{2^{i+1}},\quad a_i=
\begin{cases}
 0 & \text{if $\alpha(X_i)=\alpha(\neg X_i)=\false$,} \\
 1 & \text{if $\alpha(X_i)\not=\alpha(\neg X_i)$,} \\
 2 & \text{if $\alpha(X_i)=\alpha(\neg X_i)=\true$,}
\end{cases}\]
where $p$ is the expected payoff for player~1 from vertex~$\phi$. By the
construction of $\calG_\phi$, we have $p>0$, and the equality only holds
if $p=1$ and $a_i=1$ for all $i=1,\dots,n$, which proves that
$\alpha$ is an assignment.

Finally, we claim that $\alpha$ is a satisfying assignment. If this were not
the case, there would exist a clause~$C$ such that player~1's expected payoff
from vertex~$C$ is 0 and therefore $p<1$, where $p$ is defined as above. This
is a contradiction to the fact that $p=1$, as we have shown above.
\end{proof}

\begin{remark*}
The reduction in the proof of \cref{thm:np-hardness} can be modified
to demonstrate \NP-hardness of the threshold-free variants of \PosNE and
\StatNE, albeit at the expense of adding one more player to the game.
\end{remark*}

It follows from \cref{thm:pos-ne-np,thm:np-hardness} that \PosNE is
\NP-complete. For \StatNE, we have provided
an \NP lower bound and a \PSpace upper bound, but the exact complexity of the
problem remains unclear. Towards gaining more insight into the problem
\StatNE, we relate its complexity to the complexity of the
\emph{Square Root Sum Problem} (\SqrtSum), the problem of deciding, given
numbers $d_1,\dots,d_n,k\in\bbN$, whether $\sum_{i=1}^n \sqrt{d_i}\geq k$.
Recently, it was shown that \SqrtSum belongs to the 4th level of the
\emph{counting hierarchy} \cite{AllenderBKM06}, which is a slight improvement
over the previously known \PSpace upper bound.
However, it is an open question since the 1970s whether \SqrtSum falls into
the polynomial hierarchy \cite{GareyGJ76,EtessamiY07}. We identify a
polynomial-time reduction from \SqrtSum to \StatNE.\footnote{Some authors
define \SqrtSum with $\leq$ instead of $\geq$. With this definition, we would
reduce from the complement of \SqrtSum instead.}
Hence, \StatNE is at least as hard as \SqrtSum, and showing that
\StatNE resides inside the polynomial hierarchy would imply a major
breakthrough in understanding the complexity of numerical computation.

\begin{theorem}
$\SqrtSum$ is polynomial-time reducible to \StatNE.
\end{theorem}
\begin{proof}
Given an instance $(d_1,\dots,d_n,k)$ of \SqrtSum, we construct an
SSMG~$(\calG,v_0)$ played by players $0,1,2,3$ (with
payoffs given in this order) such that $\sum_{i=1}^n
\sqrt{d_i}\geq k$ iff $(\calG,v_0)$ has a stationary Nash
equilibrium where player~0 wins almost surely.
%This alone does not establish a reduction from
%$\overline{\SqrtSum}$ since we allow equality.
%However, since it is decidable in polynomial time whether $\sum_{i=1}^n
%\sqrt{d_i}=k$ \cite{BorodinFHT85}, our construction can easily be modified
%to establish such a reduction.

In order to state our reduction, let us first examine the game $\calG(p)$,
where $p\in[\smallfrac{1}{2},1)$, which is depicted in
\cref{fig:sqrt-reduction}\,(b).
\begin{figure}
\centering
\begin{tikzpicture}[x=1.4cm,y=1.3cm,->,bend angle=15]
\begin{scope}
\node (init) at (-1.7,0) {};
\node (0) at (-1,0) [play,label={above:3}] {$v_0$};
\node (t) at (-1,-1) {$(0,0,0,\frac{2k+1}{8dn})$};
\node (1) at (0,0) [prob] {$v_1$};
\node (2) at (0,1.5) {$(1,0,0,0)$};
\node (3) at (55:1.5) {$\calG(p_1)$};
\node (4) at (-55:1.5) {$\calG(p_n)$};
\node (dots) at (0.75,0.1) {$\vdots$};

\draw (init) to (0);
\draw (0) to (1); \draw (0) to (t);
\draw (1) to node[left] {$\frac{1}{4dn}$} (2);
\draw (1) to node[right] {$\frac{4d^2-d_1}{4d^2-n}$} (3);
\draw (1) to node[right] {$\frac{4d^2-d_n}{4d^2-n}$} (4);

\node (caption) at (0.15,-2.4) {(a)};
\end{scope}

\begin{scope}[xshift=3.9cm]
\node (label) at (-0.5,1.3) {$\calG(p)$:};
\node (0) at (-0.7,0) {};
\node (1) at (0,0) [play,label=above left:{1}] {$s$};
\node (1a) at (1.1,0) {$(1,\smallfrac{1}{2},0,0)$};
\node (2) at (0.5,1) [prob] {};
\node (2a) at (0.3,2) {$(1,0,0,1)$};
\node (3) at (1.7,1) [play,label=right:{0}] {$s_1$};
\node (3a) at (1.9,2) {$(1,1,0,0)$};
\node (4) at (2.2,0) [play,label=above right:{2}] {};
\node (4a) at (3.3,0) {$(1,0,\smallfrac{1}{2},0)$};
\node (5) at (1.7,-1) [prob] {};
\node (5a) at (1.9,-2) {$(1,0,0,0)$};
\node (6) at (0.5,-1) [play,label=left:{0}] {$s_2$};
\node (6a) at (0.3,-2) {$(1,0,1,0)$};

\draw (0) to (1);
\draw (1) to [bend left] (2); \draw (1) to (1a);
\draw (2) to [bend left=20] node[above] {$p$} (3); \draw (2) to node[right] {$1-p$} (2a);
\draw (3) to [bend left] (4); \draw (3) to (3a);
\draw (4) to [bend left] (5); \draw (4) to (4a);
\draw (5) to [bend left=20] node[above] {$p$} (6); \draw (5) to node[right] {$1-p$} (5a);
\draw (6) to [bend left] (1); \draw (6) to (6a);

\node (caption) at (1.5,-2.4) {(b)};
\end{scope}
\end{tikzpicture}
\caption{\label{fig:sqrt-reduction} Reducing \SqrtSum to \StatNE.}
\end{figure}

\begin{claim}
The maximal payoff player~3 can receive in a stationary Nash equilibrium of
$(\calG(p),s)$ is $\frac{\sqrt{2-2p}-p+1}{2p+2}$.
\end{claim}
\begin{proof}
Let $\vec{\sigma}$ be any stationary strategy profile of $(\calG(p),s)$. We
denote by~$x_1$ and~$x_2$ the probabilities that
player~0 stays inside the gadget at vertex~$s_1$ and vertex~$s_2$
respectively. Consequently, the probabilities of \emph{eventually leaving}
the gadget at vertex~$s_1$ and vertex~$s_2$ are given by $p_1(x_1,x_2):=
\frac{p(1-x_1)}{1-x_1x_2p^2}$ and $p_2(x_1,x_2):=\frac{p(1-x_2)}{1-x_1x_2p^2}$
respectively.
Note that if $x_1=0$, then $\vec{\sigma}$ is a Nash equilibrium
where player~3 receives payoff $\leq 1-p\leq \frac{\sqrt{2-2p}-p+1}{2p+2}$.
Hence, let us assume that $x_1>0$ and look for a Nash equilibrium where
player~3 receives payoff~$>1-p$. For this, it must be the case that
$p_1(x_1,x_2),p_2(x_1,x_2)\geq\smallfrac{1}{2}$ since otherwise player~1 or
player~2
could improve her payoff by moving out of the gadget, where they would get
payoff \smallfrac{1}{2}\ immediately (and player~3 would receive
payoff~$\leq 1-p$).
Vice versa, if $p_1(x_1,x_2),p_2(x_1,x_2)\geq\smallfrac{1}{2}$ then
$\vec{\sigma}$ is obviously a Nash equilibrium. Hence, to determine the maximum
payoff for player~3 in a stationary Nash equilibrium, we have to
maximise $\frac{1-p}{1-x_1 x_2 p^2}$, the expected payoff for player~3, under
the constraints $p_1(x_1,x_2),p_2(x_1,x_2)\geq\smallfrac{1}{2}$ and
$0\leq x_1,x_2\leq
1$. We claim that the maximum is reached only if $x_1=x_2$; if, for example,
$x_1>x_2$ then we can achieve a higher payoff for player~3
by setting $x_2':=x_1$, and the constraints are still satisfied:
\[\frac{p(1-x_2')}{1-x_1x_2'p^2}=\frac{p(1-x_1)}{1-x_1x_2'p^2}
 =\frac{p(1-x_1)}{1-x_1^2p^2}\geq\frac{p(1-x_1)}{1-x_1 x_2p^2}
 \geq\frac{1}{2}\]
Hence, in fact, we have to maximise $\frac{1-p}{1-x^2p^2}$ under the
constraints $\frac{p(1-x)}{1-x^2p^2}\geq\frac{1}{2}$ and $0\leq x\leq 1$,
\ie under $p^2x^2-2px+2p-1\geq 0$ and $0\leq x\leq 1$. The roots of
$p^2x^2-2px+2p-1$ are $\frac{1\pm\sqrt{2-2p}}{p}$, but
$\frac{1+\sqrt{2-2p}}{p}$ is always greater than 1 for $p\in[0,1)$.
Hence, any solution must be less than $x:=\frac{1-\sqrt{2-2p}}{p}$. In
fact, we always have $0\leq x<1$ for $p\in(\smallfrac{1}{2},1)$. Therefore,
$x$ is the optimal solution, and the maximal payoff for player~3 is indeed
$\frac{1-p}{1-x^2p^2}=\frac{\sqrt{2-2p}-p+1}{2p+2}$.\noqed
\end{proof}

Finally, we can setup our reduction. Let $(d_1,\dots,d_n,k)$ be an instance of
\SqrtSum where, without loss of generality, $n>0$, $d_i>0$ for each
$i=1,\dots,n$, and $k\leq d:=\sum_{i=1}^n d_i$.
Define $p_i:= 1 - \frac{d_i}{2d^2}$ for $i=1,\dots,n$.
Note that $p_i\in [\smallfrac{1}{2},1)$ since $0<d_i\leq d\leq d^2$. For
the reduction, we use $n$ copies of the game $\calG(p)$, where in the $i$th
copy we set $p$ to $p_i$. The complete game $\calG$ is depicted in
\cref{fig:sqrt-reduction}\,(a); it can obviously be constructed in polynomial
time.

By the above claim, the maximal payoff player~3 can get in a stationary Nash
equilibrium of $(\calG(p_i),s)$ is
\[\frac{\sqrt{2-2p_i}-p_i+1}{2p_i+2}
=\frac{\frac{1}{d}\sqrt{d_i}-(1-\frac{d_i}{2d^2})+1}{4-\frac{d_i}{d^2}}
=\frac{d\sqrt{d_i}+\frac{d_i}{2}}{4d^2-d_i}\,.\]
Consequently, the maximal payoff player~3 can get in a stationary Nash
equilibrium of $(\calG,v_1)$ is
\[\sum_{i=1}^n \frac{4d^2-d_i}{4d^2n}\cdot
\frac{d\sqrt{d_i}+\frac{d_i}{2}}{4d^2-d_i}
=\sum_{i=1}^n\frac{\sqrt{d_i}}{4dn}+\sum_{i=1}^n \frac{d_i}{8d^2n}
=\sum_{i=1}^n\frac{\sqrt{d_i}}{4dn}+\frac{1}{8dn}\,.\]
Let us fix a stationary Nash equilibrium $\vec{\sigma}$ of $(\calG,v_1)$
with this payoff for player~3.

Now, if $\sum_{i=1}^n \sqrt{d_i}\geq k$, then also $\sum_{i=1}^n
\frac{\sqrt{d_i}}{4dn}+\frac{1}{8dn}\geq\frac{2k+1}{8dn}$, and $\vec{\sigma}$
can be extended to a stationary Nash equilibrium of $(\calG,v_0)$ where
player~0 wins almost surely by setting $\vec{\sigma}(v_1\mid v_0)=1$.
On the other hand, if $\sum_{i=1}^n\sqrt{d_i}<k$, then also $\sum_{i=1}^n
\frac{\sqrt{d_i}}{4dn}+\frac{1}{8dn}<\frac{2k+1}{8dn}$, and in every stationary
Nash equilibrium of $(\calG,v_0)$ player~3 leaves the game at $v_0$, which
gives payoff~0 to player~0.
\end{proof}

\section{Undecidable variants of \NE}\label{section:undecidability}

\subsection{Pure-strategy equilibria}

In this section, we show that the problem \PureNE is undecidable
by exhibiting a reduction from an undecidable problem about
\emph{two-counter machines}. Our construction is inspired by a
construction used by Br\'azdil \ea \cite{BrazdilBFK06} to prove the
undecidability of stochastic games with branching-time winning conditions.

A two-counter machine $\calM$ is given by a list of instructions $\iota_1,
\dots,\iota_m$ where each instruction is one of the following:
\begin{itemize}
 \item \Inc{j}{k}\quad (increment counter~$j$ by~1 and go to
   instruction number~$k$);
 \item \Test{j}{k}{l}\quad (if the value of
   counter~$j$ is zero, go to instruction number~$k$; otherwise, decrement
   counter~$j$ by one and go to instruction number~$l$);
 \item \Halt\quad (stop the computation).
\end{itemize}
Here $j$ ranges over $1,2$ (the two counters), and $k\not=l$ range over
$1,\dots,m$. A configuration of $\calM$ is a triple $C=(i,c_1,c_2)\in
\{1,\dots,m\}\times\bbN\times\bbN$, where~$i$ denotes the number of the
current instruction and $c_j$ denotes the current value of counter~$j$.
A configuration $C'$ is the \emph{successor} of configuration $C$, denoted by
$C\vdash C'$, if it results from $C$ by executing instruction $\iota_i$; a
configuration $C=(i,c_1,c_2)$ with $\iota_i=\Halt$ has no successor
configuration. Finally, the \emph{computation of $\calM$} is the unique maximal
sequence $\rho=\rho(0)\rho(1)\dots$ such that $\rho(0)\vdash\rho(1)\vdash\dots$
and $\rho(0)=(1,0,0)$ (the \emph{initial configuration}). Note that $\rho$ is
either infinite, or it ends in a configuration $C=(i,c_1,c_2)$ such that
$\iota_i=\Halt$.

The \emph{halting problem} is to decide, given a machine $\calM$,
whether the computation of $\calM$ is finite. It is well-known that
two-counter machines are Turing powerful, which makes the halting problem
and its dual, the \emph{non-halting problem}, undecidable.

\begin{theorem}\label{thm:undecidability}
\PureNE is undecidable.
\end{theorem}

In order to prove \cref{thm:undecidability}, we show that one can compute from
a two-counter machine~$\calM$ an SSMG~$(\calG,v_0)$ with nine players such that
the computation of $\calM$ is infinite iff $(\calG,v_0)$ has a pure Nash
equilibrium where player~0 wins almost surely. This establishes a reduction
from the non-halting problem to \PureNE.

The game~$\calG$ is played player~0 and eight other
players $A_j^t$ and $B_j^t$, indexed by $j\in\{1,2\}$ and $t\in
\{0,1\}$. Let $\Gamma=\{\init,\inc(j),\dec(j),\zero(j):j=1,2\}$. If $\calM$
has instructions $\iota_1,\dots,\iota_m$, then for each
$i\in\{1,\dots,m\}$, each $\gamma\in\Gamma$, each $j\in\{1,2\}$
and each $t\in\{0,1\}$, the game~$\calG$ contains the gadgets
$S_{i,\gamma}^t$, $I_{i,\gamma}^t$ and $C_{j,\gamma}^t$, which are depicted in
\cref{fig:two-counter-simulation}. In the figure, squares represent terminal
vertices (the edge leading from a terminal vertex to itself being
implicit), and the labelling indicates which players win at the respective
vertex. Moreover, the dashed edge inside $C_{j,\gamma}^t$ is present iff
$\gamma\not\in\{\init,\zero(j)\}$. The initial vertex~$v_0$ of $\calG$ is the
black vertex inside the gadget~$S_{1,\init}^0$.
\begin{figure}
\centering
\begin{tikzpicture}[x=1.2cm,y=1.1cm,->]
\node (caption) at (-1,1) [anchor=west] {$S_{i,\gamma}^t$:};

\begin{scope} % Main gadget
\node (1) at (0,0) [play,fill=black,label=left:$A_1^t$] {};
\node (1a) at (0.7,0)
[coordinate,label=right:{$(0,\smallfrac{1}{3},\dots,\smallfrac{1}{3})$}] {};
\node (2) at (0,-1) [play,label=left:$B_1^t$] {};
\node (2a) at (0.7,-1)
[coordinate,label=right:{$(0,\smallfrac{1}{6},\dots,\smallfrac{1}{6})$}] {};
\node (3) at (0,-2) [play,label=left:$A_2^t$] {};
\node (3a) at (0.7,-2)
[coordinate,label=right:{$(0,\smallfrac{1}{3},\dots,\smallfrac{1}{3})$}] {};
\node (4) at (0,-3) [play,label=left:$B_2^t$] {};
\node (4a) at (0.7,-3)
[coordinate,label=right:{$(0,\smallfrac{1}{6},\dots,\smallfrac{1}{6})$}] {};
\node (5) at (0,-4.5) [prob] {};
\node (6) at (0.7,-4) [coordinate,label=right:{$C_{1,\gamma}^t$}] {};
\node (7) at (0.7,-5) [coordinate,label=right:{$C_{2,\gamma}^t$}] {};
\node (8) at (0,-5.6) [coordinate,label=below:{$I_{i,\gamma}^t$}] {};

\draw (0,0.8) to (1);
\draw (1) to (2); \draw (1) to (1a);
\draw (2) to (3); \draw (2) to (2a);
\draw (3) to (4); \draw (3) to (3a);
\draw (4) to (5); \draw (4) to (4a);
\draw (5) to node[left] {\smallfrac{1}{2}} (8);
\draw (5) to node[above] {\smallfrac{1}{4}} (6);
\draw (5) to node[below] {\smallfrac{1}{4}} (7);
\end{scope}

\node (caption) at (3,0.6) [anchor=west] {$C_{j,\gamma}^t$:};

\begin{scope}[xshift=4cm,yshift=-2.64cm] % Counter gadget (increment)
\node (2) at (1,0) [play,label={below:0}] {};
\node (2a) at (1,1) [prob,fill=black!20] {};
\node (2b) at (1,2) [end,label={above:$0,A_j^t,B_j^{\vec{t}}$}] {};
\node (3) at (2,0) [prob] {};
\node (3a) at (2,1) [end,label={above:$0,A_j^t,A_j^{\vec{t}}$}] {};
\node (4) at (3,0) [prob] {};
\node (4a) at (3,1) [end,label={above:$0,A_j^t,B_j^{\vec{t}}$}] {};
\node (5) at (4,0) [prob] {};
\node (5a) at (4,1) [end,label={above:$0,A_j^t,B_j^{\vec{t}}$}] {};
\node (6) at (5,0) [end,label={above:$0,B_j^t,B_j^{\vec{t}}$}] {};
\node (label) at (0,-1) [anchor=west] {if $\gamma=\inc(j)$;};

\draw (0.3,0) to (2);
\draw (2) to (3); \draw (2) to [bend right] (2a);
\draw (2a) to (2b); \draw (2a) to [bend right] (2);
\draw (3) to (4); \draw (3) to (3a);
\draw (4) to (5); \draw (4) to (4a);
\draw (5) to (6); \draw (5) to (5a);
\end{scope}

\begin{scope}[xshift=4cm,yshift=-7.24cm] % Counter gadget (decrement)
\node (2) at (1,0) [play,label={below:0}] {};
\node (2a) at (1,1) [prob,fill=black!20] {};
\node (2b) at (1,2) [end,label={above:$0,A_j^t,B_j^{\vec{t}}$}] {};
\node (3) at (2,0) [prob] {};
\node (3a) at (2,1) [end,label={above:$0,A_j^t,B_j^{\vec{t}}$}] {};
\node (4) at (3,0) [prob] {};
\node (4a) at (3,1) [end,label={above:$0,A_j^t,B_j^{\vec{t}}$}] {};
\node (5) at (4,0) [prob] {};
\node (5a) at (4,1) [end,label={above:$0,A_j^t,A_j^{\vec{t}}$}] {};
\node (6) at (5,0) [end,label={above:$0,B_j^t,B_j^{\vec{t}}$}] {};
\node (label) at (0,-1) [anchor=west] {if $\gamma=\dec(j)$;};

\draw (0.3,0) to (2);
\draw (2) to (3); \draw (2) to [bend right] (2a);
\draw (2a) to (2b); \draw (2a) to [bend right] (2);
\draw (3) to (4); \draw (3) to (3a);
\draw (4) to (5); \draw (4) to (4a);
\draw (5) to (6); \draw (5) to (5a);
\end{scope}

\begin{scope}[xshift=4cm,yshift=-11.84cm] % Counter gadget (nochange)
\node (2) at (1,0) [play,label={below:0}] {};
\node (2a) at (1,1) [prob,fill=black!20] {};
\node (2b) at (1,2) [end,label={above:$0,A_j^t,B_j^{\vec{t}}$}] {};
\node (3) at (2,0) [prob] {};
\node (3a) at (2,1) [end,label={above:$0,A_j^t,B_j^{\vec{t}}$}] {};
\node (4) at (3,0) [prob] {};
\node (4a) at (3,1) [end,label={above:$0,A_j^t,A_j^{\vec{t}}$}] {};
\node (5) at (4,0) [prob] {};
\node (5a) at (4,1) [end,label={above:$0,A_j^t,B_j^{\vec{t}}$}] {};
\node (6) at (5,0) [end,label={above:$0,B_j^t,B_j^{\vec{t}}$}] {};
\node (label) at (0,-1) [anchor=west] {if $\gamma\not\in\{\inc(j),\dec(j)\}$.};

\draw (0.3,0) to (2);
\draw (2) to (3); \draw (2) to [dashed,bend right] (2a);
\draw (2a) to (2b); \draw (2a) to [bend right] (2);
\draw (3) to (4); \draw (3) to (3a);
\draw (4) to (5); \draw (4) to (4a);
\draw (5) to (6); \draw (5) to (5a);
\end{scope}

\node (caption) at (-1,-6.5) [anchor=west] {$I_{i,\gamma}^t$:};

\begin{scope}[yshift=-8cm] % Instruction gadget (Increment)
\node (1) at (0,0) [play,label={above:0}] {};
\node (2) at (0.7,0) [coordinate,label=right:{$S_{k,\inc(j)}^{\vec{t}}$}] {};
\node (label) at (-1,-0.8) [anchor=west]
{if $\iota_i=\Inc{j}{k}$;};

\draw (-0.7,0) to (1);
\draw (1) to (2);
\end{scope}

\begin{scope}[yshift=-10.1cm] % Instruction gadget (Test))
\node (1) at (0,0) [play,label={above:0}] {};
\node (2) at (0.7,0.5) [coordinate,label=right:{$S_{k,\zero(j)}^{\vec{t}}$}]
{};
\node (3) at (0.7,-0.5) [coordinate,label=right:{$S_{l,\dec(j)}^{\vec{t}}$}]
{};
\node (label) at (-1,-1.2) [anchor=west]
{if $\iota_i=\Test{j}{k}{l}$;};

\draw (-0.7,0) to (1);
\draw (1) to (2);
\draw (1) to (3);
\end{scope}

\begin{scope}[yshift=-12.5cm] % Instruction gadget (Halt)
\node (1) at (0,0) [play,label={above:0}] {};
\node (2) at (0.7,0) [coordinate,label=right:{$(0,\dots,0)$}] {};
\node (label) at (-1,-0.8) [anchor=west] {if $\iota_i=\Halt$.};

\draw (-0.7,0) to (1);
\draw (1) to (2);
\end{scope}

\end{tikzpicture}
\caption{\label{fig:two-counter-simulation} Simulating a two-counter machine.}
\end{figure}
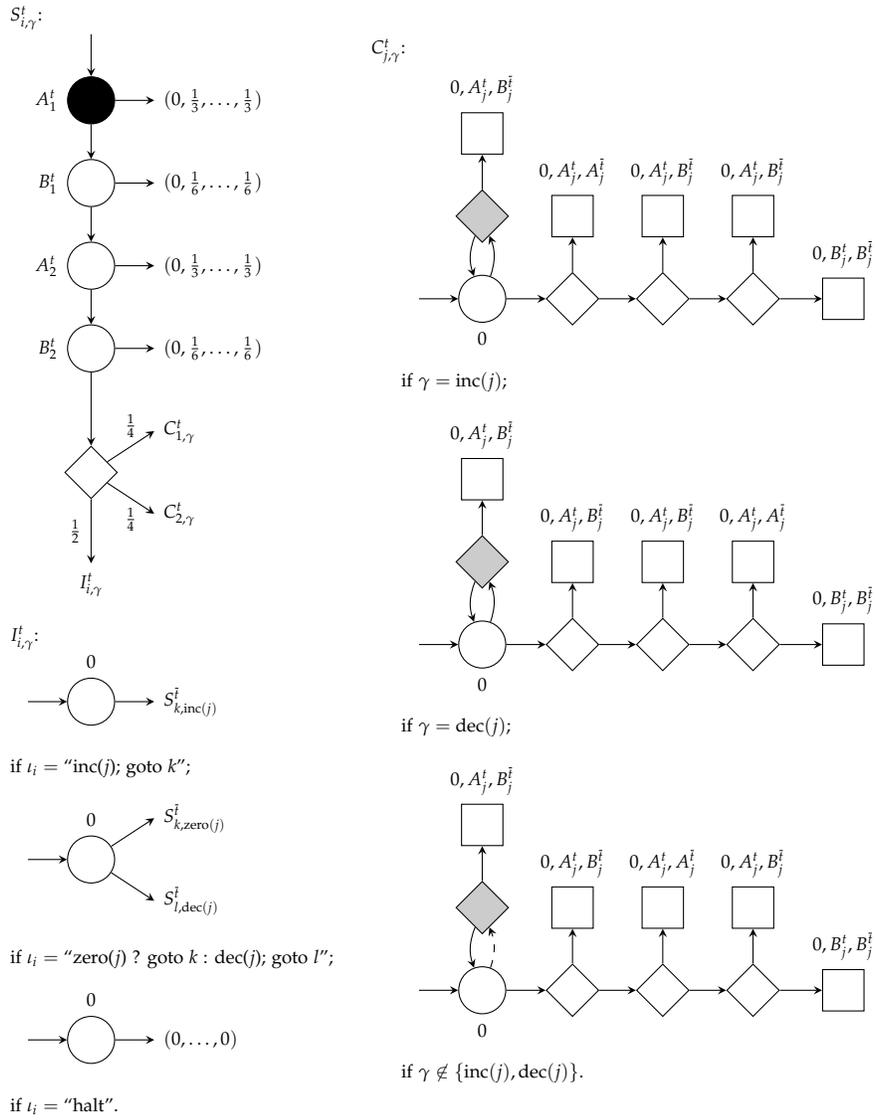

For any pure strategy profile $\vec{\sigma}$ of $\calG$ where player~0 wins
almost surely, let $x_0v_0\prec
x_1v_1\prec x_2v_2\prec \dots$ ($x_i\in V^*,v\in V$, $x_0=\epsilon$) be the
(unique) sequence of all consecutive histories such that, for each $n\in\bbN$,
$v_n$ is a black vertex and
$\Prob_{v_0}^{\vec{\sigma}}(x_n v_n\cdot V^\omega)>0$.
Additionally, let $\gamma_0,\gamma_1,\dots$
be the corresponding sequence of instructions, \ie $\gamma_n=\gamma$ for the
unique instruction~$\gamma$ such that $v_n$ lies in one of the
gadgets~$S_{i,\gamma}^t$ (where $t=n\bmod 2$).
For each $j\in\{1,2\}$ and $n\in\bbN$, we define two conditional probabilities
$a_j^n$ and $p_j^n$ as follows:
\begin{align*}
a_j^n &:=\Prob_{v_0}^{\vec{\sigma}}(\Reach(F_{A_j^{n\bmod 2}})\mid
 x_n v_n\cdot V^\omega) \\
\shortintertext{and}
p_j^n &:=\Prob_{v_0}^{\vec{\sigma}}(\Reach(F_{A_j^{n\bmod 2}})\mid
 x_n v_n\cdot V^\omega\setminus x_{n+2} v_{n+2}\cdot V^\omega).
\end{align*}
Finally, for each $j\in\{1,2\}$ and $n\in\bbN$, we define an ordinal number
$c_j^n\leq\omega$ as follows: After the history $x_n v_n$, with
probability~\smallfrac{1}{8}\ the play proceeds to the vertex controlled by
player~0 in the counter gadget~$C_{j,\gamma_n}^t$ (where $t=n\bmod 2$). The
number $c_j^n$ is
defined to be the maximal number of subsequent visits to the grey vertex inside
this gadget (where $c_j^n=\omega$ if, on one path, the grey
vertex is visited infinitely often). Note that, by the construction of
$C_{j,\gamma}^t$, it holds that $c_j^n=0$ if $\gamma_n=\zero(j)$ or
$\gamma_n=\init$.

\begin{lemma}\label{lemma:reduction}
Let $\vec{\sigma}$ be a pure strategy profile of $(\calG,v_0)$ where
player~0 wins almost surely. Then $\vec{\sigma}$ is a Nash equilibrium if
and only if
\begin{align}\label{equation:counter-update}
c_j^{n+1} &=
\begin{cases}
1+c_j^n & \text{if $\gamma_{n+1}=\inc(j)$,} \\
c_j^n-1 & \text{if $\gamma_{n+1}=\dec(j)$,} \\
c_j^n=0 & \text{if $\gamma_{n+1}=\zero(j)$,} \\
c_j^n & \text{otherwise}
\end{cases}
\end{align}
for all $j\in\{1,2\}$ and $n\in\bbN$.
\end{lemma}

Here $+$ and $-$ denote the usual addition and subtraction of ordinal numbers
respectively (satisfying $1+\omega=\omega-1=\omega$).
The proof of \cref{lemma:reduction} goes through several claims. In the
following, let $\vec{\sigma}$ be a pure strategy profile of $(\calG,v_0)$
where player~0 wins almost surely. The first claim
gives a necessary and sufficient condition on the probabilities~$a_j^n$ for
$\vec{\sigma}$ to be a Nash equilibrium.

\begin{claim*}
The profile $\vec{\sigma}$ is a Nash equilibrium iff $a_j^n=
\smallfrac{1}{3}$ for all $j\in\{1,2\}$ and $n\in\bbN$.
\end{claim*}
\begin{proof}
($\Rightarrow$) Assume that $\vec{\sigma}$ is a Nash equilibrium. Clearly,
this implies that
$a_j^n\geq\smallfrac{1}{3}$ for all $n\in\bbN$
since otherwise some player~$A_j^t$ could improve her payoff by leaving one of
the gadgets~$S_{i,\gamma}^t$. Let
\[b_j^n:=\Prob_{v_0}^{\vec{\sigma}}(\Reach(F_{B_j^{n\bmod 2}})\mid
 x_n v_n\cdot V^\omega).\]
We have $b_j^n\geq\smallfrac{1}{6}$ for all $n\in\bbN$
since otherwise some player~$B_j^t$ could improve her payoff by leaving one of
the gadgets~$S_{i,\gamma}^t$. Note that
at every terminal vertex of the counter gadgets $C_{j,\gamma}^t$ and
$C_{j,\gamma}^{\vec{t}}$ either player~$A_j^t$ or player~$B_j^t$ wins.
The conditional probability that, given the history $x_n v_n$, we reach
one of those gadgets is
$\sum_{k\in\bbN}\frac{1}{2^k}\cdot\frac{1}{4}=\frac{1}{2}$
for all $n\in\bbN$, so we have $a_j^n=\smallfrac{1}{2}-b_j^n$ for all
$n\in\bbN$. Since $b_j^n\geq\smallfrac{1}{6}$, we arrive at $a_j^n
\leq\smallfrac{1}{2}-\smallfrac{1}{6}=\smallfrac{1}{3}$, which proves the
claim.

($\Leftarrow$) Assume that $a_j^n=\smallfrac{1}{3}$ for all $n\in\bbN$.
Clearly, this implies that none of the players~$A_j^t$ can improve her
payoff. To show that none of the players~$B_j^t$ can improve her payoff, it
suffices to show that $b_j^n\geq\smallfrac{1}{6}$ for all $n\in\bbN$. But
with the same argumentation as above, we have $b_j^n=\smallfrac{1}{2}-
a_j^n$ and thus $b_j^n=\smallfrac{1}{6}$ for all $n\in\bbN$, which proves
the claim.
\end{proof}

\noindent
The second claim relates the probabilities~$a_j^n$ and $p_j^n$.

\begin{claim*}
Let $j\in\{1,2\}$. Then $a_j^n=\smallfrac{1}{3}$ for all $n\in\bbN$ if and only
if $p_j^n=\smallfrac{1}{4}$ for all $n\in\bbN$.
\end{claim*}
\begin{proof}
($\Rightarrow)$ Assume that $a_j^n=\smallfrac{1}{3}$ for all $n\in\bbN$.
We have $a_j^n=p_j^n+\smallfrac{1}{4}\cdot a_j^{n+2}$ and therefore
$\smallfrac{1}{3}=p_j^n+\smallfrac{1}{12}$ for all $n\in\bbN$. Hence,
$p_j^n=\smallfrac{1}{4}$ for all $n\in\bbN$.

($\Leftarrow$) Assume that $p_j^n=\smallfrac{1}{4}$ for all $n\in\bbN$.
Since $a_j^n=p_j^n+\smallfrac{1}{4}\cdot a_j^{n+2}$ for all $n\in\bbN$,
the numbers~$a_j^n$ must satisfy the following recurrence: $a_j^{n+2}=
4a_j^n-1$. Since all the numbers~$a_j^n$ are
probabilities, we have $0\leq a_j^n\leq 1$ for all $n\in\bbN$. It is easy
to see that the only values for $a_j^0$ and $a_j^1$
such that $0\leq a_j^n\leq 1$ for all $n\in\bbN$ are $a_j^0=a_j^1=
\smallfrac{1}{3}$. But this implies that $a_j^n=\smallfrac{1}{3}$ for all
$n\in\bbN$.
\end{proof}

\noindent
Finally, the last claim relates the numbers $p_j^n$ to
\cref{equation:counter-update}.

\begin{claim*}
Let $j\in\{1,2\}$. Then $p_j^n=\smallfrac{1}{4}$ for all $n\in\bbN$ if and only
if \cref{equation:counter-update} holds for all $n\in\bbN$.
\end{claim*}
\begin{proof}
Let $n\in\bbN$, and let $t=n\bmod 2$. The probability $p_j^n$ can be
expressed as the sum of the probability that the play reaches a terminal vertex
that is winning for player~$A_j^t$ inside $C_{j,\gamma_n}^t$ and the
probability that the play reaches such a vertex inside
$C_{j,\gamma_{n+1}}^{\vec{t}}$.
The first probability does not depend on $\gamma_n$, but the second depends on
$\gamma_{n+1}$. Let us consider the case that $\gamma_{n+1}=\inc(j)$. In this
case, the aforementioned sum is equal to the following sum of two binary
numbers:
\[0.00\underbrace{1\dots 1}_{\rlap{\scriptsize $c_j^n$ times}}111 +
0.000\underbrace{0\dots 0}_{\rlap{\scriptsize $c_j^{n+1}$ times}}100\,.\]
Obviously, this sum is equal to \smallfrac{1}{4}\ iff $c_j^{n+1}=1+c_j^n$.
For any other value of $\gamma_{n+1}$, the argumentation is similar, and we
omit it here.
\end{proof}

\begin{proof}[Proof of \cref{lemma:reduction}]
By the first claim, the profile
$\vec{\sigma}$ is a Nash equilibrium iff $a_j^n=\smallfrac{1}{3}$ for all
$j\in\{1,2\}$ and $n\in\bbN$.
By the second claim, the latter is true if $p_j^n=\smallfrac{1}{4}$
for all $j\in\{1,2\}$ and $n\in\bbN$. Finally, by the last claim, this is the
case iff \cref{equation:counter-update} holds for all $j\in\{1,2\}$ and
$n\in\bbN$.
\end{proof}

To establish the reduction, it remains to show that the computation of $\calM$
is infinite iff the game $(\calG,v_0)$ has a pure Nash equilibrium where
player~0 wins almost surely.

($\Rightarrow$) Assume that the computation $\rho=\rho(0)\rho(1)\dots$ of
$\calM$ is infinite. We define a pure strategy~$\sigma_0$ for player~0 as
follows: For a history that ends in one of the instruction
gadgets~$I_{i,\gamma}^t$ after visiting a black vertex exactly $n$~times,
player~0 tries to move to the neighbouring gadget~$S_{k,\gamma'}^{\vec{t}}$
such that $\rho(n)$ refers to instruction number~$k$ (which is always possible
if $\rho(n-1)$ refers to instruction number~$i$; in any other case, $\sigma_0$
might be defined arbitrarily). In particular, if $\rho(n-1)$ refers to
instruction $\iota_i=\Test{j}{k}{l}$, then player~0 will move to the
gadget~$S_{k,\zero(j)}^{\vec{t}}$ if the value of the counter in configuration
$\rho(n-1)$ is~0 and to the gadget~$S_{l,\dec(j)}^{\vec{t}}$ otherwise.
For a history that ends in one of the gadgets~$C_{j,\gamma}^t$ after
visiting a black vertex exactly $n$~times and a grey vertex exactly $m$~times,
player~0 will move to the grey vertex again iff $m$ is strictly
less than the value of the counter~$j$ in configuration $\rho(n-1)$. So
after entering $C_{j,\gamma}^t$, player~0's strategy is to loop through the
grey vertex exactly as many times as given by the value of the counter~$j$ in
configuration~$\rho(n-1)$.

Any other player's pure strategy is ``moving down at any time''.
We claim that the resulting strategy profile~$\vec{\sigma}$ is a Nash
equilibrium of $(\calG,v_0)$ where player~0 wins almost surely.

Since, according to her strategy, player~0 follows the computation of $\calM$,
no vertex inside an instruction gadget~$I_{i,\gamma}^t$ where $\iota_i$ is the
halt instruction is ever reached. Hence, with probability~1 a terminal vertex
in one of the counter gadgets is reached. Since player~0 wins at any such
vertex, we can conclude that she wins almost surely.

It remains to show that $\vec{\sigma}$ is a Nash equilibrium. By the definition
of player~0's strategy~$\sigma_0$, we have the following for all~$n\in\bbN$:
1.\ $c_j^n$ is the value of counter~$j$ in configuration
$\rho(n)$; 2.\ $c_j^{n+1}$ is the value of counter~$j$ in configuration
$\rho(n+1)$; 3.\ $\gamma_{n+1}$ is the instruction corresponding to the counter
update from configuration $\rho(n)$ to $\rho(n+1)$. Hence,
\cref{equation:counter-update} holds, and $\vec{\sigma}$ is a Nash equilibrium
by \cref{lemma:reduction}.

$(\Leftarrow)$ Assume that $\vec{\sigma}$ is a pure Nash equilibrium of
$(\calG,v_0)$ where player~0 wins almost surely. We define an infinite sequence
$\rho=\rho(0)\rho(1)\dots$ of \emph{pseudo configurations} (where the counters
may take the value~$\omega$) of $\calM$ as follows.
Let $n\in\bbN$, and assume that $v_n$ lies inside the
gadget~$S_{i,\gamma_n}^t$ (where $t=n\bmod 2$); then $\rho(n):=
(i,c_1^n,c_2^n)$.

We claim that $\rho$ is, in fact, the (infinite) computation of $\calM$. It
suffices to verify the following two properties:
\begin{enumerate}
 \item $\rho(0)=(1,0,0)$;
 \item $\rho(n)\vdash\rho(n+1)$ for all $n\in\bbN$.
\end{enumerate}
Note that we do not have to show explicitly that each $\rho(n)$ is a
configuration of $\calM$ since this follows easily by induction from 1.\ and
2. Verifying the first property is easy: $v_0$ lies inside $S_{1,\init}^0$
(and we are at instruction~1), which is linked to the counter gadgets
$C_{1,\init}^0$ and $C_{2,\init}^0$. The edge leading to the grey vertex is
missing in these gadgets. Hence, $c_1^0$ and $c_2^0$ are both equal to~0.

For the second property, let $\rho(n)=(i,c_1,c_2)$ and $\rho(n+1)=(i',c_1',
c_2')$. Hence, $v_n$ lies inside $S_{i,\gamma}^t$ and $v_{n+1}$ inside
$S_{i',\gamma'}^{\vec{t}}$ for suitable $\gamma,\gamma'$ and $t=n\bmod 2$. We
only proof the claim for the case that $\iota_i=\Test{2}{k}{l}$; the other
cases are straightforward. Note that, by the construction of the
gadget~$I_{i,\gamma}^t$, it must be the case that either $i'=k$ and
$\gamma'=\zero(2)$, or $i'=l$ and $\gamma'=\dec(2)$. By \cref{lemma:reduction},
if $\gamma'=\zero(2)$, then $c_2'=c_2=0$ and $c_1'=c_1$, and if
$\gamma'=\dec(2)$, then $c_2'=c_2-1$ and $c_1'=c_1$. This implies $\rho(n)
\vdash\rho(n+1)$: On the one hand, if $c_2=0$, then $c_2'\not=c_2-1$, which
implies $\gamma'\not=\dec(2)$ and thus $\gamma'=\zero(2)$, $i'=k$ and $c_2'=
c_2=0$. On the other hand, if $c_2>0$, then $\gamma'\not=\zero(2)$ and thus
$\gamma'=\dec(2)$, $i'=l$ and $c_2'=c_2-1$.\qed

\subsection{Finite-state equilibria}

It follows from the proof of \cref{thm:undecidability} that Nash equilibria may
require infinite memory (even if we are only interested in whether a player wins
with probability~0 or 1). More precisely, we have the following proposition.

\begin{proposition}\label{prop:inf-mem}
There exists an SSMG that has a pure Nash equilibrium where player~0 wins
almost surely but that has no finite-state Nash equilibrium where player~0
wins with positive probability.
\end{proposition}
\begin{proof}
Consider the game~$(\calG,v_0)$ constructed in the proof of
\cref{thm:undecidability} for the machine~$\calM$ consisting of the single
instruction \Inc{1}{1}. We modify this game by adding a new initial
vertex~$v_1$ which is controlled by a new player, player~1, and from where she
can either move to~$v_0$ or to a new terminal vertex where she receives
payoff~1 and every other player receives payoff~0. Additionally, player~1 wins
at every terminal vertex of the game~$\calG$ that is winning for player~0. Let
us denote the modified game by $\calG'$.

Since the computation of $\calM$ is infinite, the game~$(\calG,v_0)$ has a pure
Nash equilibrium where player~0 wins almost surely. This equilibrium induces
a pure Nash equilibrium of~$(\calG',v_1)$ where player~0 wins almost surely.

Now assume that there exists a finite-state Nash equilibrium of $(\calG',v_1)$
where player~0 wins with positive probability. Such an equilibrium induces a
finite-state Nash equilibrium of $(\calG,v_0)$ where player~1, and thus also
player~0, wins almost surely since otherwise player~1 would play to $v_0$ with
probability~1. By \cref{lemma:reduction}, this implies that player~0 updates
the counter correctly. However, since player~0 uses a finite-state strategy,
the corresponding counter values are bounded by a constant, a contradiction.
\end{proof}

Note that \FinNE is recursively enumerable: To decide whether an
SSMG~$(\calG,v_0)$ has a finite-state Nash equilibrium with payoff
$\geq\vec{x}$ and $\leq\vec{y}$, one can just enumerate all possible
finite-state profiles and check for each of them whether the profile is a Nash
equilibrium with
the desired properties by analysing the finite Markov chain that is generated by
this profile (where one identifies states that correspond to the same vertex
and memory state). Hence, to show the undecidability of \FinNE, we cannot
reduce from the non-halting problem but from the halting problem for
two-counter machines (which is recursively enumerable itself).

\begin{theorem}\label{thm:undecidability-fin-ne}
\FinNE is undecidable.
\end{theorem}
\begin{proof}
The construction is similar to the one for proving undecidability of \PureNE.
Given a two-counter machine~$\calM$, we modify the SSMG~$\calG$ constructed in
the proof of \cref{thm:undecidability} by
adding another ``counter'' (together with four more players for checking
whether the counter is updated correctly) that has to be
incremented in each step. Moreover, additionally to the terminal vertices in
the gadgets~$C_{j,\gamma}^t$, we let player~0 win at the terminal vertex
in each of the gadgets~$I_{i,\gamma}$ where $\iota_i=\Halt$. Let us
denote the new game by $\calG'$. Now, if $\calM$ does not halt, any pure Nash
equilibrium of $(\calG',v_0)$ where player~0 wins almost surely needs
infinite memory: to win almost surely, player~0 must follow the computation of
$\calM$ and increment the new counter at each step. On the other hand, if
$\calM$ halts, then we can easily construct a finite-state Nash equilibrium of
$(\calG',v_0)$ where player~0 wins almost surely. Hence, $(\calG',v_0)$ has a
finite-state Nash equilibrium where player~1 wins almost surely iff the
machine~$\calM$ halts. The details of the construction are left to the reader.
\end{proof}

\section{Conclusion}

We have analysed the complexity of deciding whether a simple stochastic
multiplayer game has a Nash equilibrium whose payoff falls into a
certain interval.
Our results demonstrate that the presence of both stochastic vertices and
more than two players makes the problem much more complicated than when one of
these factors is absent. In particular, the problem of deciding the existence
of a pure-strategy Nash equilibrium where player~0 wins almost surely is
undecidable for simple stochastic multiplayer games, whereas it is contained in
$\NP\cap\coNP$ for two-player, zero-sum simple stochastic games \cite{Condon92}
and even in \PTime for non-stochastic infinite multiplayer games with, \eg,
B\"uchi winning conditions \cite{Ummels08}.

Apart from settling the complexity of \NE when arbitrary mixed strategies are
considered, future research may, for example, investigate restrictions of \NE
to games with a small number of players.
%\footnote{Recall that our undecidability proof for
%pure strategies uses nine players.}
In particular, we conjecture that the problem is decidable for two-player
games, even if these are not zero-sum.

\bibliographystyle{gbabbrv}
\bibliography{../../biblio/all.bib}

\end{document}